\documentclass[12pt]{article}

\usepackage{amsthm,amscd,amssymb}

\usepackage{amsmath}
\usepackage{amsfonts}
\usepackage{graphicx}
\usepackage{psfrag}

\newcounter{mnotecount}%[section]

\newcommand{\mnotex}[1]%{}
{\protect{\stepcounter{mnotecount}}$^{\mbox{\footnotesize $\bullet$\themnotecount}}$ 
\marginpar{%\color{red}%
\raggedright\tiny\em
$\!\!\!\!\!\!\,\bullet$\themnotecount: #1} }

\newcommand{\bm}[1]{\mbox{\boldmath $#1$}}

\DeclareFontFamily{OT1}{rsfs}{}
\DeclareFontShape{OT1}{rsfs}{m}{n}{ <-7> rsfs5 <7-10> rsfs7 <10->
rsfs10}{} \DeclareMathAlphabet{\mycal}{OT1}{rsfs}{m}{n}

\def\defi{:=}

\def\Id{\mbox{\bm{Id}}}
\def\Kend{\mbox{\bm{K}}}

\def\gE{g_E}
\def\cnull{W}
\def\T{T}

\def\Ten{\mathcal{T}}
\def\D{D}
\def\C{\bm{C}}
\def\grad{\mbox{grad}}
\def\VS{V_S}
\newtheorem{Tma}{Theorem}
\newtheorem{Prop}{Proposition}

\newtheorem{Lem}{Lemma}

\theoremstyle{remark} 
\def\M{{\cal M}}
\def\Sbar{\overline{S}}

\def\Cod{D \, \sffuno}
\def\Hess{\mbox{Hess}}
\def\tr{\mbox{tr}}

\def \Mcal {(\mathcal{M}^{1,n+1},\eta)}

\def \grmetuno {\gamma}
\def \grmetdos {\overline{\gamma}}
\def \sffuno {K}
\def \sffdos {\overline{K}}
\def \gruno {S}
\def \grdos {\overline{S}}
\def \vecuno {X}
\def \vecdos {\overline{X}}
\def \vecunoy {Y}
\def \vecdosy {\overline{Y}}
\def \cuno {c}
\def \cdos {\overline{c}}
\def \nuno {\nu}
\def \ndos {\overline{\nu}}
\def \tvecdos {\widetilde{X}}
\def \tvecdosy {\widetilde{Y}}
\def \tndos{\widetilde{\nu}}

\def \del {\bm{\delta}}
\def \nabgen {\nabla^{\mathcal{M}}}
\def \W {W}

\def \metdos {\gamma}
\def \mettres {\overline{\gamma}}

\def \supdos {S}
\def \suptres {\overline{S}}
\def \smink {\bm{\mathcal{S}}}

\def \paral {\mathcal{T}_{p\rightarrow \psi(p)}}
\def \paralinv {\mathcal{T}_{\psi(p)\rightarrow p}}

\def\JournalPrep#1#2#3{#1, ``#2'', #3.}
\def\Journal#1#2#3#4#5#6{#1, ``#2'', {\em #3} {\bf #4}, #5 (#6).}

\def\CQG{\em Class. Quantum Grav.}
\def\JPA{\em J. Phys. A: Math. Gen.}

\def\CMP{\em Commun. Math. Phys.}

\def\PRL{\em Phys. Rev. Lett.}

\def\ANYAS{\em Ann. N. Y. Acad. Sci.}

\def\AM{\em Ann. Math.}

\begin{document}

%Y ahora ponemos el estilo en cabecera y a pie de pagina donde queramos
%\renewcommand{\headrulewidth}{1pt}
%\renewcommand{\footrulewidth}{1pt}

%Con el * desaparece el numero de los capitulos, y tambien vale para secciones y subsecciones

\title{
Geometry of normal graphs in Euclidean space and applications
to the Penrose inequality in Minkowski}

\author{Marc Mars$^1$ and Alberto Soria$^2$ \\
Facultad de Ciencias, Universidad de Salamanca,\\
Plaza de la Merced s/n, 37008 Salamanca, Spain \\
$^1$also at Instituto de F\'{\i}sica Fundamental y Matem\'aticas, \\
Universidad de Salamanca,  \, marc@usal.es,  $^2$asoriam@usal.es }

\maketitle

%Esto solo sirve para el esrilo libro    \chapter{Trabajamos  matematicas}

\begin{abstract}
The Penrose inequality in Minkowski is a geometric
inequality relating the total outer
null expansion and the area of closed, connected
and spacelike codimension-two surfaces $\smink$
in the Minkowski spacetime, subject to
an additional convexity assumption. In a recent paper,
Brendle and Wang \cite{BrendleWang2013} 
find a sufficient condition for the validity of 
this Penrose inequality in terms   
of the geometry of the orthogonal projection of $\smink$ onto a constant time hyperplane. In this work, we study the geometry of hypersurfaces
in $n$-dimensional euclidean space which 
are normal graphs over other surfaces  and relate the intrinsic and
extrinsic geometry of the graph with that of the base hypersurface.
These results are used to rewrite Brendle and Wang's condition 
explicitly in terms of the time height function of $\smink$
over a hyperplane
and the geometry of the projection of $\smink$ along its past null cone
onto this hyperplane.     
We also include, in an Appendix, a self-contained 
summary of known and new results
on the geometry of projections along the Killing direction  
of codimension two-spacelike surfaces in 
a strictly static spacetime.
\end{abstract}

\section{Introduction}

The {\it Penrose inequality in Minkowski} refers to a geometric inequality
for a class of co\-di\-men\-sion-two spacelike surfaces $\smink$ embedded in the
$(n+2)$-dimensional Minkowski spacetime $(\M^{1,n+1}, \eta)$. The surfaces
are restricted to be closed, connected, orientable and spacetime convex in the sense
that their second fundamental form along one of its future directed null
normals is  non-positive\footnote{Our sign conventions are such that the 
second fundamental form of sphere in $\mathbb{R}^3$ with respect to the outer
normal is positive definite.}.
It follows  \cite{MarsSoria2012} that this property
can only happen for one such 
future null direction. Indeed, if both future null second fundamental forms
were non-positive then the mean curvature vector $H$ of the surface
would be future causal and not-identically vanishing (because closed codimension-two surfaces in Minkowski cannot
be totally geodesic). Thus, $\smink$ would be future trapped and not minimal,
which cannot occur, e.g. by the results in \cite{MarsSenovilla2003}. The (unique)
future directed null normal for which the second fundamental form is non-positive will be denoted by $k$ (note that this field is defined up to an arbitrary
positive scaling) and referred as the future directed {\it inner} null normal.
The future directed  {\it outer} null normal $\ell$ is defined by the conditions of
being null, orthogonal to $\smink$, future directed and satisfying $\langle k,
\ell \rangle_{\eta} = -2$, where $\langle \cdot, \cdot \rangle_{\eta}$ is the
scalar product with the Minkowski metric $\eta$. Then the Penrose
inequality in Minkowski can 
be written as 
\begin{eqnarray}
\int_{\smink} - \langle H, \ell \rangle_{\eta} \langle k,\xi \rangle_{\eta} \bm{\eta_{\mathcal{S}}}
\geq n ( \omega_n)^{\frac{1}{n}} |\smink|^{\frac{n-1}{n}}, \label{PI}
\end{eqnarray}
where $\omega_n$ is the total area of a the unit $n$-sphere, 
$\bm{\eta_{\mathcal{S}}}$ is the induced measure on $\smink$, $|\smink|$ its total area
and $\xi$ is any choice of future directed
unit timelike Killing vector in Minkowski (referred
from now on as a time translation). Note that there is not just
one Penrose inequality in Minkowski, but one for
each choice of time translation $\xi$. When a distinction is necessary, we will refer to "the Penrose inequality with respect to $\xi$".

The physical motivation for this inequality comes from a construction
due to Penrose \cite{Penrose1973}
where an incoming null shell of dust matter propagates in
the Minkowski spacetime along the null geodesics with tangent vector $k$. The null 
hypersurface $\Omega$ they sweep is smooth all the way from $\smink$ to past
null infinity  
(here is where the condition of spacetime convexity becomes important).
Then, the standard Penrose inequality relating total (Bondi) energy and area
of trapped surfaces can be rewritten in terms of Minkowskian 
quantities only. Details on the construction can be found, for instance, in
\cite{Penrose1973,Mars2009,Tod1985,Tod1998} .

 Despite its apparent simplicity, inequality (\ref{PI}) is still
a difficult open problem.
It was proved by Gibbons \cite{Gibbons1997} in the case of surfaces
lying on the hyperplane orthogonal to $\xi$, where it becomes precisely the classic 
Minkowski inequality for convex surfaces in Euclidean space. For surfaces
lying in the past null cone of a point (denoted as ``spherical case'',
although obviously the surface $\smink$ is not in general spherically symmetric),
the inequality was proved by Tod \cite{Tod1985} in spacetime dimension four by using suitable 
Sobolev inequalities and extended to arbitrary 
spacetime dimension (bigger than three) in \cite{MarsSoria2012} as a consequence
of the Beckner inequality for spheres \cite{Beckner1993}. In fact, the 
spherical case can also be viewed
as a particular case of a Penrose inequality for spacetimes admitting
shear-free null hypersurfaces extending from the trapped surface to past null infinity proved by Sauter \cite{Sauter2008}. 

In spacetime dimension four, the inequality (\ref{PI})
has been established for a large class of surfaces
\cite{MarsSoria2012} using a  
geodesic flow of surfaces along $\Omega$ starting on $\smink$ and
adapting and extending  previous ideas of 
Ludvigsen and Vickers \cite{LudvigsenVickers1983} and Bergqvist \cite{Bergqvist1997}.
Wang \cite{Wang2012} has proved the inequality for
surfaces lying on a spacelike hyperboloid of Minkowski with the properties of being mean convex and star-shaped with respect to the point of tangency of the
hyperboloid with the foliation by constant time hyperplanes orthogonal
to $\xi$. Very recently Brendle and Wang \cite{BrendleWang2013} 
have proved
the inequality for another large class of
surfaces, namely those lying on a timelike cylindrical hypersurface
with generator $\xi$ and base a convex surface in a constant time
hyperplane orthogonal to $\xi$. These cylinders are called
{\it convex static timelike hypersurfaces} in \cite{BrendleWang2013}. 
In fact, the case analyzed by the authors refers to a generalization
of inequality (\ref{PI}) conjectured for the Schwarzschild
spacetime, but the argument applies to the Minkowski situation.
The main idea behind their result consists in performing a projection of 
$\smink$ along the time translation $\xi$ onto a constant time hyperplane
$\Sigma_{t_0}$. 
By relating the geometry of $\smink$  to the geometry of the projected surface
$\Sbar$  on $\Sigma_{t_0}$,
inequality (\ref{PI}) becomes a consequence of the standard Minkowski
inequality in Euclidean space provided $\Sbar$ is convex.
More precisely, Brendle and Wang prove the following
result (the result is established in \cite{BrendleWang2013} in spacetime
dimension four, but the argument is in fact dimensional independent):

%\vspace{5mm}

\begin{Tma}[S. Brendle $\&$ M.T. Wang]
\label{ThWang}
Let $\Mcal$ be the (n+2)-dimensional Min\-kow\-ski spacetime with $t$ a Minkowskian time defining a unit Killing $\pmb{\xi}=-dt$. 
Let $\smink$ be a closed, connected, orientable and
spacetime convex surface  in $\Mcal$ 
with contravariant metric $\gamma^{-1}$.
Let $\pi: \M^{1,n+1} \rightarrow \Sigma_{t_0}$ be
the orthogonal projection onto the hyperplane $\Sigma_{t_0}=\{t=t_0\}$
and define $\suptres = \pi(\smink)$. Denote by
$\bm{\eta_{\suptres}}$ its volume form and by $\overline{K}$ 
its second fundamental form as a hypersurface
of $(n+1)$-Euclidean space
with respect to the outer unit normal.
Then the Penrose inequality with respect to $\xi$ for $\smink$ is equivalent to
\begin{eqnarray}
& &\int_{\suptres} \tr_{d \pi (\gamma^{-1})}
\overline{K} \,  \bm{\eta_{\suptres}}\geq n(\omega_n)^\frac{1}{n}|\smink|^\frac{n-1}{n}
\end{eqnarray}
and  holds if $\overline{S}$ is convex. 
\end{Tma}

\vspace{5mm}

The aim of this paper is to analyze this case in further detail by
writing out the condition of $\Sbar$ being convex explicitly in terms of the
height function of $\smink$ (defined below) and the geometry of the 
convex surface $S$ obtained by intersecting the null hypersurface $\Omega$
ruled by the null geodesics starting at on $\smink$ with tangent vector $-k$ and
the hyperplane $\Sigma_{t_0}$. This requires analyzing the geometry
of $\Sbar$ as a graph over $S$. Although a purely Euclidean
calculation we have not been able to find the result in the literature
and most of our work consists in relating the induced metric and second
fundamental forms of $\Sbar$ to those of $S$. We devote Section \ref{SECglobalconstruction}
to present these calculations and the consequences they have on the
Penrose inequality on convex static timelike hypersurfaces. Our main result 
is Theorem \ref{T0}, where the explicit differential inequality that
the height function of $\smink$ needs to satisfy for a surface $\smink$ to 
lie on a convex static timelike hypersurface is obtained.

Our  second aim in this paper consists in presenting in a concise and unified
manner the geometric relationship between the geometry of the codimension-two surface 
$\smink$ and its projection $\Sbar$,
which lies at the core
of the argument by Brendle and Wang. Some of these results have already
appeared in several places in the literature, but in a somewhat scattered
manner and not in a completely  exhaustive form.
Given the potential usefulness for such ``vertical'' projection in other
areas of physics and geometry, we believe it to be convenient to 
present all the results in a unified and complete manner. We do this for an
arbitrary strictly static spacetime in the Appendix.

\section{Geometry of normal graphs on hypersurfaces of $\mathbb{E}^{n+1}$}
\label{SECGRAPH}

The following conventions and notation are used: if $S$ is any embedded 
spacelike submanifold in a semi-Riemannian manifold, we denote by
$\gamma$ and $D$ the induced metric and corresponding  covariant derivative.
%Tensors in $S$ will carry Latin capital indices. 
The second fundamental form and mean curvature vectors 
are $\vec{K}(X,Y)\defi -(\nabla_{X}Y)^\bot$ and
$H \defi \mbox{tr}_{\gamma} K$, where $X,Y$ are tangent vector field to $S$ and
$\bot$ denotes the normal component to $S$. 
If $\nu$ is a vector field
orthogonal to $S$, the extrinsic curvature
along $\nu$ is $K^{\nu}(X,Y)\defi \langle \nu,\vec{K}(X,Y)\rangle$.
All manifolds and tensors are assumed to be smooth.

In this section, the ambient manifold is 
the $(n+1)$-Euclidean space $(\mathbb{E}^{n+1},\gE)$, $n \geq 2$. 
The flat connection is denoted by $\nabla$ and the corresponding 
(global) parallel transport by 
$\mathcal{T}_{p_1\rightarrow p_2}:T_{p_1}\mathbb{E}^{n+1} \longrightarrow 
T_{p_2}\mathbb{E}^{n+1}$, for any  $p_1,p_2\in\mathbb{E}^{n+1}$. Obviously,
in Cartesian
coordinates  $\{x^\alpha\}$, ($\alpha=1,\dots,n+1$) this map simply preserves
the coefficients of any vector $V \in T_{p_1} \mathbb{E}^{n+1}$ 
in the basis $\{\partial_{x^{\alpha}}\}$.  $\{ x^{\alpha} \}$ will always refer
to a Cartersian coordinate system. 

Consider two embedded submanifolds $\supdos$ and $\suptres$ in
$\mathbb{E}^{n+1}$ and assume there is diffeomorphism
$\psi: \supdos \longrightarrow \suptres$.
The following result relates the tangential
covariant derivative of vector fields along $\supdos$ (not necessarily tangent to $\supdos$) with the corresponding parallely transported vector field on
$\suptres$. This result will play an important role
below.
\begin{Lem}
\label{lemrelcampostransportados}
Let $\supdos$, $\suptres$ and $\psi$ as above. Let $Z$ be a vector field along $\supdos$. Consider $\vecuno$ a vector field tangent to $\supdos$
and define $\mathcal{T}Z |_{\psi(p)} \defi\mathcal{T}_{p\rightarrow\psi(p)}Z_p$ $\forall p\in \supdos$. Then 
\begin{equation}
\label{paralleldifferenciation}
\paral(\nabla_{\vecuno} Z|_{p})=(\nabla_{d\psi(\vecuno)}\mathcal{T} Z)|_{\psi(p)}, 
\end{equation} 
\end{Lem}

\begin{proof}
The left-hand side of (\ref{paralleldifferenciation}) is
\begin{equation}
\label{leftside}
\paral(\nabla_{\vecuno}Z|_p)=\paral\left(\vecuno(Z^\alpha)\partial_{x^\alpha}|_p\right)=\vecuno(Z^\alpha) |_{p} \partial_{x^\alpha}|_{\psi(p)}. 
\end{equation}
On the other hand, on $\suptres$ we have
$\mathcal{T}Z=Z^{*\alpha}\partial_{x^\alpha}$, where $Z^{*\alpha}=
Z^{\alpha}\circ\psi^{-1}$. 
Viewing $Z^{\alpha}$ as scalar functions 
we can also write $Z^{*\alpha} = (\psi^{-1})^{\star} (Z^{\alpha})$.
Its covariant derivative along $d\psi|_{p}(\vecuno)$ is
\begin{align}
\label{rightside}
(\nabla_{d\psi(\vecuno)}\mathcal{T} Z)|_{\psi(p)} & =\nabla_{d\psi(\vecuno)}
\left(Z^{*\alpha}\partial_{x^\alpha}\right)|_{\psi(p)}=
d\psi(\vecuno) ( (\psi^{-1})^{\star} (Z^{\alpha})) \partial_{x^\alpha}|_{\psi(p)} 
%=dZ^{*\alpha}_{\psi(p)}(\vecdos)\partial_{x^\alpha}|_{\psi(p)} \nonumber \\
%& &=d(Z^{\alpha}\circ\psi^{-1})_{\psi(p)}(d\psi_p(\vecuno))
%\partial_{x^\alpha}|_{\psi(p)}=dZ^\alpha_p(\vecuno)
%\partial_{x^\alpha}|_{\psi(p)}
& = \vecuno(Z^\alpha) |_p \partial_{x^\alpha}|_{\psi(p)}
\end{align}
which is the same as (\ref{leftside}). \end{proof}

Assume now that $\gruno$ is an orientable hypersurface and select
a unit normal vector field $\nuno$. Choose a smooth
function $\sigma : \gruno \longrightarrow \mathbb{R}$ and consider
the set of points at signed distance $\sigma$ from each $p \in \gruno
\subset \mathbb{E}^{n+1}$
along the normal $\nuno(p)$. The congruence of normal geodesics
to $\gruno$ meets no focal points for distances 
$\sigma$ satisfying the bound
\begin{equation}
\label{sigmabound}
|\sigma| < \frac{1}{\underset{1\leq A\leq n}{\max} \{|\kappa_A |\}},
\end{equation}
where 
$\{ \kappa_A \} $ are the principal curvatures of $\gruno$. Assuming this
bound from now on, we have that the map
$\psi' : \gruno \rightarrow \mathbb{E}^{n+1}$ defined by
\begin{eqnarray}
\psi'(p) = p + \sigma(p) \nuno(p), \label{psi'}
\end{eqnarray}
(where we are obviously using the affine structure of 
$\mathbb{E}^{n+1}$) is such that $\grdos \defi \psi'(S)$ 
is an embedded hypersurface of Euclidean space, and, in fact, a graph
over $\gruno$.  Our aim is to relate the induced metrics
and second fundamental forms of $\gruno$ and $\grdos$.

It is clear that the restriction of $\psi'$ onto its image
is a diffeomorphism between $\gruno$ and $\grdos$, which will be denoted by
$\psi$. Let $\vecuno \in \mathfrak{X} (\gruno)$ be a vector 
field tangent to $\gruno$ 
and define $\vecdos \defi d\psi (\vecuno)$, which is obviously tangent
to $\grdos$.

For the purposes of this section, it is convenient to
transport paralelly $\vecdos$ from $\psi(p)$ to $p$ because this will
allow us to perform all calculations in a single manifold.
Thus, let us define the vector field $\tvecdos \in \mathfrak{X} (\gruno)$
as $\tvecdos |_p \defi \paralinv (\vecdos |_{\psi(p)} )$. The first aim is to
relate  $\tvecdos$ with $\vecuno$. Consider any
curve $\cuno(s)$ in $\gruno$ passing through $p \in \gruno$ with tangent
vector $\vecuno |_p$. From the definition of $\psi$, the curve
$\cdos \defi \psi \circ \cuno$ has tangent vector at $\psi(p)$
given by $\paral (\vecuno + \sigma\nabla_{\vecuno} \nuno |_p+ X(\sigma) \nuno )$.
Recalling that the Weingarten map $\Kend : T_p \gruno \longrightarrow 
T_p \gruno$ is defined by
$\Kend (\vecuno) \defi \nabla_{\vecuno} \nuno$ we conclude
\begin{equation}
\label{relattan}
\tvecdos|_p=(\vecuno+\sigma \nabla_{\vecuno}\nuno + \vecuno(\sigma) \nuno)|_{p}
=\left (\Id+\sigma\Kend + d\sigma \otimes \nuno \right ) (\vecuno) )|_p.
\end{equation}
From the geometric construction of $\grdos$ it is intuitively clear that 
the normal vector $\ndos$ orthogonal to $\grdos$ must satisfy
$\gE ( \tndos|_p, \nuno|_p ) \neq 0$ for all $p \in \gruno$, where 
$\tndos|_p \defi \paralinv (\ndos |_{\psi(p)})$. For a rigorous
proof we use (\ref{relattan}) as follows.
Given that $\paral$ is an isometry, (\ref{relattan}) implies
the following identity, valid for any $\vecuno \in T_p \gruno$:
\begin{align}
0 & = \gE ( \ndos, \vecdos) |_{\psi(p)} =
\gE ( \paral (\tndos), \paral(\tvecdos))|_{\psi(p)} =
\gE ( \tndos , \tvecdos ) |_p = \nonumber \\
&  =\gE ( \tndos , (\Id+ \sigma \Kend) (\vecuno)  )|_p
 +d \sigma (\vecuno) \gE (\tndos, \nuno) |_p . \label{prod}
\end{align}
Assume there is $ p \in \gruno$ such that $\tndos|_p \in T_p \gruno$
(i.e. $\gE(\tndos,\nuno)=0$).
Then $\gE ( \tndos , (\Id+ \sigma \Kend) (\vecuno)  )|_p =0$
for any $\vecuno \in T_p \gruno$, which is a contradiction
with the fact that the bound (\ref{sigmabound})  
implies that the endomorphism 
$\Id + \sigma \Kend$ is invertible. 

Let us choose the orientation
of $\ndos$ so that $\cnull \defi  \gE ( \tndos,\nuno) > 0$ on $\gruno$.
Thus, we can decompose $\tndos = \cnull (\nuno - \T)$ on $\gruno$,
where $\T \in \mathfrak{X} (\gruno)$ is a tangent vector field.
Equation (\ref{prod})  implies
\begin{equation}
\label{coeftan}
\T =  ( \Id + \sigma \Kend )^{-1} ( \grad_{\,\gamma}(\sigma) ),
\end{equation}
where $\grad_{\,\gamma}(\sigma)$ is the gradient of $\sigma$ with respect
to the induced metric
$\grmetuno$.  For notational simplicity, define the invertible endomorphism 
$\C  \defi (\Id + \sigma \Kend)$  so that 
$\T = \C^{-1} (\grad_{\,\gamma}(\sigma))$. 
The condition of $\ndos$ being unit fixes $\cnull$
to satisfy $\cnull^2 (1 + \grmetuno(\T,\T)) = 1$, which, given our choice
of normal in $\grdos$, implies
\begin{equation}
\label{coefnorm}
\cnull= \frac{1}{\sqrt{ 1 + \grmetuno(\T,\T)}}.
\end{equation}
We are ready to prove our main result of this section, which 
relates the geometry of the graph $\grdos$ with the geometry of its base 
$\gruno$.
\begin{Tma}
\label{tmarelcurvaturas}
Consider the hypersurfaces $\gruno$, $\grdos$ of Euclidean
space $(\mathbb{E}^{n+1},\gE)$ with signed distance function
$\sigma$ and diffeomorphism
$\psi$, as above. The respective induced metrics
$\grmetuno$ and $\grmetdos$ and second fundamental forms
$\sffuno$ and $\sffdos$ with respect to the  normals $\nuno$ and $\ndos$ are related by
\begin{align}
\label{metricenunc}
 \psi^{\star} (\grmetdos) & =\grmetuno +2\sigma\sffuno 
+ \sigma^2 \sffuno \circ \sffuno + d \sigma \otimes d \sigma, 
\\
\frac{1}{\cnull} \psi^{\star} (\sffdos) & = 
\sffuno + \sigma \sffuno \circ \sffuno + \sigma \Cod (\cdot,
\T, \cdot )  + d \sigma \otimes \sffuno(\T, \cdot ) 
+ \sffuno(\T, \cdot ) \otimes d \sigma - 
\Hess_{\,\gamma}(\sigma), \label{sffrelation}
%& &\frac{1}{\cnull}\sffdos_{AB}=\sffuno_{AB}-\widehat{D}_A \widehat{D}_B\sigma+\sigma \left( \sffuno_{AL}\sffuno^L_{\phantom{L}B}+b^L \widehat{D}_A (\sffuno_{BL}) \right)+b^L \left(\sigma_{,A}\sffuno_{BL}+\sigma_{,B}\sffuno_{AL}\right), \nonumber \\
\end{align}
where $\T$ and $\cnull$ are defined in (\ref{coeftan})-(\ref{coefnorm}),
$\sffuno \circ \sffuno$ is the trace
of $\sffuno\otimes \sffuno$ in the second and third indices,
%(\vecuno_1,\vecuno_2) \defi
%\grmetuno(\Kend \circ \Kend (\vecuno_1), \vecuno_2)$,
%$\Cod (\vecuno_1,\vecuno_2,\vecuno_3) \defi (\D_{\vecuno_1} \sffuno)
%(\vecuno_2,\vecuno_3)$  where 
$\D$ is the Levi-Civita derivative of $\grmetuno$ and $\Hess_{\,\gamma}(\sigma)$ is the 
Hessian of $\sigma$ in this metric.
%with $\cnull$ as in (\ref{coefnorm}), $\widehat{D}$ the covariant derivative related to the metric $\grmetuno$, $b^L=\sigma_{,M}{\left[(\pmb{Id}+\sigma\pmb{\sffuno})^{-1}\right]}^{ML}$ and $\sigma$ as above described.
\end{Tma}

\noindent {\bf Remark.} These expressions reduce to well-known results
when either $\sigma$ is constant or when the base surface is a hyperplane.

\vspace{5mm}

\noindent {\bf Remark.} It is interesting that
the symmetry of $\sffdos$ for any $\sigma$ is equivalent to 
the Codazzi identity $\Cod (\vecuno_1,\cdot , \vecuno_3) = 
\Cod (\vecuno_3,\cdot,\vecuno_1)$ for $\gruno$. So, properties
of normal graphs can be used to derive curvature identities
on the base hypersurface, which  usually would require different methods.

\vspace{5mm}

\begin{proof}
Let $\vecuno, \vecunoy 
\in \mathfrak{X}(\gruno)$ be arbitrary tangent vector fields.
We start with (\ref{metricenunc}). 
With the notation above, and using that
the parallel transport is an isometry:
\begin{align*}
\psi^{\star} (\grmetdos) (\vecuno,\vecunoy) |_{\psi(p)} & = 
\grmetdos \left (d \psi|_p (\vecuno), d\psi|_p (\vecunoy) \right )
= \gE (\vecdos,\vecdosy)|_{\psi(p)} = 
\gE (\tvecdos,\tvecdosy) |_p = 
\gE (\C (\vecuno), \C(\vecunoy)) |_p + \\
& + d\sigma \otimes d\sigma|_p  
(\vecuno, \vecunoy) 
= \grmetuno ( (\Id + \sigma \Kend) (\vecuno),
(\Id + \sigma \Kend) (\vecunoy))|_p + d\sigma \otimes d\sigma|_p  
(\vecuno, \vecunoy) 
\end{align*}
where in the fourth equality we used (\ref{relattan}). 
This establishes (\ref{metricenunc}).
To prove (\ref{sffrelation}) we first apply Lemma 
\ref{lemrelcampostransportados} to find the identity
\begin{align}
\gE (\nabla_{\vecuno}\tndos, \tvecdosy  ) |_{p} & 
=\gE \left ( \paral  (\nabla_{\vecuno}\tndos),
\paral ( \tvecdosy ) \right ) |_{\psi(p)} =
\gE \left ( \nabla_{\vecdos}\ndos ,\vecdosy \right ) |_{\psi(p)} = \nonumber \\
& = \sffdos (\vecdos,\vecdosy) |_{\psi(p)} = 
\psi^{\star} (\sffdos) (\vecuno,\vecunoy) |_p. \label{inter}
\end{align}
To evaluate the left-hand side we recall the fundamental identity,
$\nabla_{\vecuno}  \vecunoy = \D_{\vecuno} \vecunoy - \sffuno (\vecuno,\vecunoy) \nuno$,  valid for any pair of tangential vector fields.
Since $\gE(\tndos,\tvecdosy)=0$, the left-hand side of (\ref{inter}) becomes
\begin{align}
\gE (\nabla_{\vecuno}\tndos, \tvecdosy  ) & = 
\frac{\vecuno (\cnull)}{\cnull}  \gE(\tndos, \tvecdosy  )  
+ \cnull \gE \left ( \nabla_{\vecuno}  ( \nuno - \T ), d\sigma (\vecunoy)
\nuno + \C (\vecunoy) \right ) = \nonumber \\
& = \cnull \grmetuno \left ( \Kend(\vecuno) - \D_{\vecuno} \T, \C(\vecunoy) \right )
+ \cnull \sffuno(\vecuno,\T ) d \sigma (\vecunoy).  \label{inter2}
\end{align}
The first term is immediately
$\cnull \gamma( \Kend(\vecuno), \C(\vecunoy) ) =
\cnull ( \sffuno + \sigma \sffuno \circ \sffuno ) (\vecuno,
\vecunoy)$.
To elaborate the second term, we use that
the endomorphism $\C$ is symmetric with
respect to $\grmetuno$, i.e. $\grmetuno(\vecuno_1,\C (\vecuno_2))=
\grmetuno(\C(\vecuno_1),\vecuno_2)$. Thus,
\begin{align*}
- \grmetuno ( \D_{\vecuno} \T,\C(\vecunoy)) & =
- \grmetuno ( ( \C \circ \D_{\vecuno} \C^{-1}) ( \grad_{\,\gamma} (\sigma) ), \vecunoy )
-  \grmetuno ( \D_{\vecuno} \grad_{\,\gamma} (\sigma), \vecunoy)  \\
& =  \grmetuno((\D_{\vecuno} \C) (\T), \vecunoy) -
\Hess_{\,\gamma}(\sigma) (\vecuno,\vecunoy)  \\
& = d \sigma (\vecuno) \sffuno (\T,\vecunoy) 
+ \sigma \Cod(\vecuno,\T,\vecunoy) -
\Hess_{\,\gamma}(\sigma) (\vecuno,\vecunoy), 
\end{align*}
where in the first equality we used $(\ref{coeftan})$ and in the second equality $ - \C \circ (\D_{\vecuno} \C^{-1})
=   (\D_{\vecuno} \C) \circ \C^{-1}$ .
Inserting this into (\ref{inter2}) yields the result.
\end{proof}

\noindent {\bf Remark.} The Riemannian character of the ambient 
Euclidean space
has only been used when evaluating $\gE (\nuno,\nuno)$ and $\gE(\ndos,\ndos)$.
With the same arguments as before, let $\gruno$ be an embedded submanifold of
the Minkowski spacetime $(\M^{1,n+1},\eta)$ with non-degenerate induced
metric $\grmetuno$ and unit normal $\nuno$ satisfying
$\langle \nuno,\nuno \rangle_{\eta} = \epsilon$ with 
$\epsilon = \pm 1$. 
$\grdos$ is constructed
as before, where the orientation of the unit normal $\ndos$ is selected so that it satisfies $\langle \tndos,\nuno\rangle_{\eta}= \epsilon \cnull$, with $W>0$.  
Under these conditions:
\begin{align}
\label{metricenunc_bis}
 \psi^{\star} (\grmetdos) & =\grmetuno +2\sigma\sffuno 
+ \sigma^2 \sffuno \circ \sffuno + \epsilon d \sigma \otimes d \sigma, 
\\
\frac{1}{\cnull} \psi^{\star} (\sffdos) & = 
\sffuno + \sigma \sffuno \circ \sffuno + \sigma \Cod (\cdot,
\T, \cdot )  + d \sigma \otimes \sffuno(\T, \cdot ) 
+ \sffuno(\T, \cdot ) \otimes d \sigma -\epsilon 
\Hess_{\,\gamma}(\sigma), \label{sffrelation_bis}
\end{align}
where all definitions are as before and the decomposition $\tndos=\cnull(\nuno-\T)$ still holds, but this time $\T$ reads $\T=\epsilon ( \Id + \sigma \Kend )^{-1} ( \grad_{\,\gamma}(\sigma) )$ and $W$ is 
\begin{equation*}
\cnull = \frac{1}{\sqrt{1 + \epsilon \grmetuno(\T,\T)}}.
\end{equation*}
 The condition $1 + \epsilon \grmetuno(\T,\T) > 0$ is necessary 
for $\grdos$ to be of the same causal character as $\gruno$.

\section{Matching two different projections}
\label{SECglobalconstruction}

Let $\Mcal$ be the $(n+2)$-dimensional Minkowski spacetime ($n \geq 2$). 
Choose a Minkowskian time $t$ so that we can define a unit Killing $\pmb{\xi}=-dt$. The constant time hyperplanes $\{t = t_0\}$ will be denoted by $\Sigma_{t_0}$. Let the codimension-two surface $\smink$ and its normal null frame $\{k,\ell\}$ be as in the Introduction. 
The convex surface 
$S \hookrightarrow \Sigma_{t_0}$ defines uniquely a null hypersurface $\Omega$ (defined {\it spacetime convex null hypersurface} in \cite{MarsSoria2012})  and, then, any spacelike
surface embedded in $\Omega$ is defined uniquely by the time height function over $\Sigma_{t_0}$, namely the function $\tau \defi 
t|_{\smink} - t_0$. This function is defined on $\smink$.
However, there is a natural diffeomorphism that maps $\smink$ to $S$ via null
geodesics in $\Omega$ tangent to $k$, so that any geometric information can be
transferred from $\smink$ onto $S$ and viceversa. This applies to any
scalar function $f$ and in particular to
$\tau$. We
will use indistinctly the same name for both functions, the precise
meaning being clear from the context.  For any closed {\it spacetime
convex} surface $\smink$, $\suptres$ must be embedded in
$\Sigma_{t_0}$ (otherwise  two different points of $\smink$ with different time heights would project the same point onto $\Sigma_{t_0}$ which is impossible
given that they lie on a smooth null hypersurface).
We can apply Theorem \ref{tmarelcurvaturas} to relate the geometry of $S$ and $\suptres$ as follows:
\begin{figure}[!htb]
\begin{center}
\psfrag{sig}{$\Sigma_{t_0}=\{t=t_0\}$}
\psfrag{min}{$\Mcal$}
\psfrag{om}{$\Omega$}
\psfrag{kill}{$\xi$}
\psfrag{s}{$\smink$}
\psfrag{s2}{$S$}
\psfrag{s3}{$\overline{S}$}
\psfrag{t1}{$\tau$}
\psfrag{t2}{$\tau$}
\psfrag{m}{$\nu$}
\psfrag{q}{$q$}
\psfrag{k}{$k$}
\psfrag{l}{$\ell$}
\psfrag{n}{$\overline{\nu}$}
\includegraphics[width=14cm]{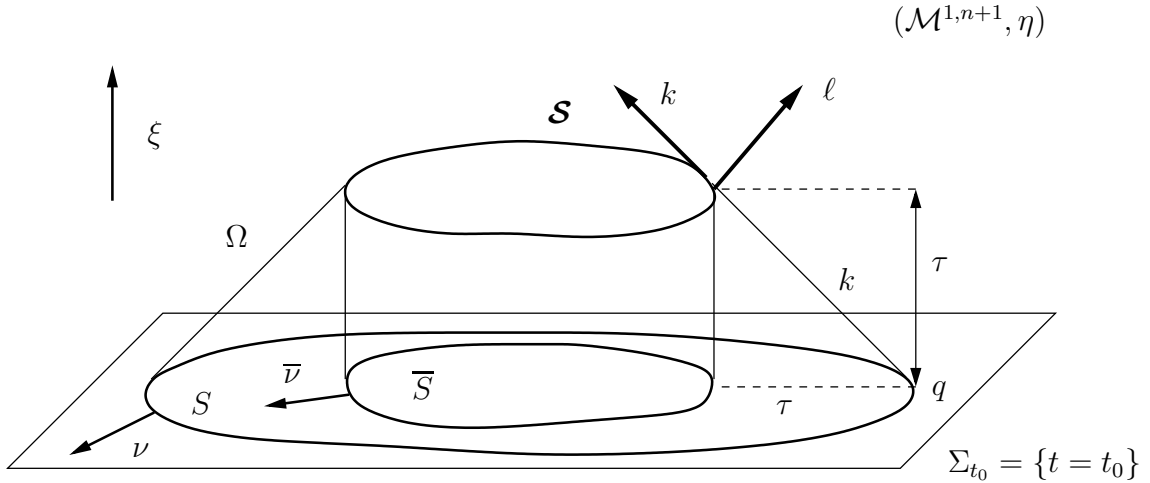}	
\end{center}
\caption{Schematic figure combining both projections: the spacetime
convex surface $\smink$ is projected along $\Omega$ onto 
$\Sigma_{t_0}$, with $S=\Omega\cap\Sigma_{t_0}$. $\suptres$ is obtained by projecting $\smink$ along the Killing $\xi$. $\{k,\ell\}$ are normalized so that $\langle k,\ell\rangle_{\eta}=-2.$ 
 }
\label{fig3}
\end{figure}

\begin{Tma}[\bf Sufficient condition for the Penrose inequality in Minkowski in terms of spacetime convex geometry]
\label{T0}
Let $\Mcal$ be the (n+2)-dimensional Minkowski spacetime with $t$ a Minkowskian time defining a unit Killing $\pmb{\xi}=-dt$. Let $\smink$ be a
closed, connected, orientable 
and spacetime convex surface in $\Mcal$ and $\Omega$ the 
convex null hypersurface containing $\smink$. Consider $S\defi\Omega\cap\Sigma_{t_0}$ and let $K$ be its second fundamental form as an euclidean surface of $\Sigma_{t_0}$ with respect to its outer unit normal $\nu$ (see Figure \ref{fig3}), $D$ the Levi-Civita connection of the metric $\metdos$ of $\supdos$, and $\grad_{\,\gamma}(\tau)$ and $\Hess_{\,\gamma}(\tau)$ the gradient and Hessian of $\tau$ in the metric $\gamma$ respectively. Let $\tau \defi t |_{\smink} - t_0$ be defined as before. If the tensor 
\begin{equation}
\label{sffrelationth}
\Ten=\sffuno - \tau \sffuno \circ \sffuno - \tau \Cod (\cdot,
\T, \cdot )  - d \tau \otimes \sffuno(\T, \cdot ) 
- \sffuno(\T, \cdot ) \otimes d \tau + 
\Hess_{\,\gamma}(\tau) 
\end{equation}
is positive semidefinite, where $\T = -( \Id - \tau \Kend )^{-1}(\grad_{\,\gamma}(\tau))$,
then the Penrose inequality with respect to $\xi$ holds for $\smink$.
\end{Tma}

\begin{proof}
Observe that in the euclidean hyperplane $\Sigma_{t_0}$ we can obtain $\suptres$ as a graph over $\supdos$ moving inwards along the inner normal to $\supdos$. Indeed, let $\nuno$ and $\ndos$ be the outer unit normals of $\supdos$ and $\suptres$. 
Moving along geodesics tangent to $k$ in the past null cone $\Omega$
 a time height $\tau$ with respect to $\Sigma_{t_0}$ is equivalent to the projected trajectory moving inwards the same signed distance $\tau$ (see Figure \ref{fig3}). Thus, we can apply Theorem \ref{tmarelcurvaturas} with $\sigma=-\tau$ and conclude that $\Ten=\frac{1}{\cnull}\psi^*(\overline{K})$ with $\cnull>0$. The validity of the Penrose inequality for $\smink$ is then a consequence of Theorem \ref{ThWang}.    
  \end{proof}

To get a flavour of the range of applicability of this result, let us consider a
few examples. Consider a closed, axially symmetric surface $\supdos$
in a spacelike hyperplane $\Sigma_{t_0}$
of four-dimensional Minkowski spacetime $\mathcal{M}^{1,3}$, and assume
that this surface is a cylinder between two parallel planes $z=z_0$
and $z=z_1$ orthogonal
to the axis of symmetry. Let $\rho_0$ be the radius of the cylinder. In
cylindrical coordinates $\{ \varphi,z\}$, (\ref{sffrelationth})
becomes, in the region $z_0 \leq  z \leq  z_1$,
%\mnotex{M:  $\rho_0$ added, check}
\begin{equation}
\label{cilindefpos}
\mathcal{T}_{AB}=(\rho_0-\tau)\del^\varphi_A\del^\varphi_B+\tau_{,AB}+\frac{\tau_{,\varphi}}{\rho_0-\tau}(\tau_{,A}\del^\varphi_B+\tau_{,B}\del^\varphi_A), 
\end{equation} 
with $\del$ the Kronecker delta. Assuming $\tau$ also axially symmetric, then
$\mathcal{T}$ is positive semidefinite if and only if $\tau_{,zz}\geq 0$. 
So, any smooth axially symmetric surface $\smink$ projecting to $S$ along the past null cone and for which $\tau$
is a constant $\tau_1$ on $z \geq z_1$, a constant $\tau_0$ on $z \leq z_0$
and fulfills $\tau_{,zz} \geq 0$ on $z\in [z_0,z_1]$,
 satisfies the Penrose inequality 
(with respect to the time translation orthogonal to the hyperplane 
$\Sigma_{t_0}$).

Another simple example is obtained when $\supdos$ is a sphere of radius 
$r_0$ in $\Sigma_{t_0}$. In spherical coordinates 
$\{\theta,\varphi\}$ (we are again in four spacetime dimensions)
non-negativity of the tensor $\mathcal{T}$ reads
%\mnotex{M: check}
\begin{equation*}
\mathcal{T}_{AB}=(r_0-\tau)\metdos_{AB}+D_A D_B \tau
+\frac{2}{r_0-\tau}\tau_{,A}\tau_{,B} \geq 0
\end{equation*}
which, in the case that $\smink$ is axially symmetric, becomes (after
adapting the spherical coordinates so that $\tau(\theta)$)
\begin{equation}
\label{convexgraphsphere}
(r_0-\tau)^2+(r_0-\tau)\tau_{,\theta\theta}+2(\tau_{,\theta})^2 \geq 0 
\quad \quad (r_0-\tau)\sin\theta+\cos\theta \tau_{,\theta}\geq 0.   
\end{equation}
Let us solve these inequalities in the strictly convex
case (i.e.
with strict inequalities in (\ref{convexgraphsphere})). With the definition 
$z(\theta)\defi(r_0 - \tau(\theta))
\cos \theta$, the second inequality becomes $z_{,\theta} < 0$, 
which can be inverted to define $\theta(z)$. With the definition
$\rho (z) \defi (r_0- \tau(\theta)) \sin \theta |_{\theta(z)}$, the first
inequality becomes, after a straightforward computation, $\rho_{,zz} < 0$.
Note also that $\rho - z \rho_{,z} = - \frac{(r_0 - \tau )^2}{z_{,\theta}} > 0$
as a consequence of their definitions. Conversely, let $\rho(z)$ satisfy
$\rho_{,zz} < 0$ and $\rho - z \rho_{,z} > 0$. Define
$z(\theta)$ by $\cos \theta = z (\sqrt{z^2 + \rho(z)^2})^{-1}|_{z=z (\theta)}$ 
(the condition $\rho - z \rho_{,z} > 0$ is used here)
and construct a function $\tau(\theta)$ be means of 
$\tau = r_0 - \sqrt{z^2 + \rho(z)^2} |_{z=z(\theta)}$. Then the surface
$\smink$ defined by this time height over the sphere $\supdos$
satisfies the Penrose inequality. 

We note that the Penrose inequality for surfaces
$\smink$ lying in the past null cone of a point in the Minkowski
spacetime has been established in full generality in \cite{Tod1985}
(for dimension 4)
and \cite{MarsSoria2012} (in any dimension). So, the second
example above does not
extend in any way the class of surfaces for which the inequality holds.
However, besides giving us an idea of the proportion of surfaces
in the null cone case covered by Theorem \ref{ThWang}, it also provides
a method to construct a wide family of axially symmetric surfaces $\smink$
for which the Penrose inequality holds. Indeed, assume now that $\supdos$
is axially symmetric and consider axially symmetric
functions $\tau$ on $\supdos$ so that
$\suptres$ is strictly convex. Let $e_z$ be
the unit field tangent to the axis of symmetry and 
$e_{\rho}$ the unit field radially outward from the axis of symmetry.
Define  the two functions on $\supdos$
\begin{equation}
\label{eqesfera}
z (p) \defi\gE(x-\tau \nu,e_z) |_p, \quad \quad \rho(p) \defi \gE(x-\tau \nu,e_\rho) |_p,
\end{equation}
where $x$ is the position vector of a point $p$ on $\supdos$ 
and $\nu$ the outward normal at $p$. The strict inequality ${\mathcal T}> 0$
is equivalent to (i) $z$ being a coordinate on $\supdos$ away from points
on the axis of symmetry and (ii) $\rho(z)$ satisfying
$\rho_{,zz} < 0$. Conversely, given any
function  $\rho(z)$ satisfying $\rho_{,zz} < 0$, if there are two maps $z,\tau :
S \rightarrow \mathbb{R}$ solving the algebraic equations (\ref{eqesfera})
with $\rho(p) \defi \rho(z(p))$, then the
spacetime surface $\smink$ defined by this time height function over $S$
satisfies the Penrose inequality. The algebraic equations will be solvable 
provided the parametric surface $\{ \rho(z),z,\varphi\}$ in cylindrical coordinates
is a normal graph over $S$. It is obvious that this is not always the case,
so restrictions are necessary. In the spherical case above, this restriction
is precisely $\rho - z \rho_z >0$.

As already mentioned,
the first case where the Penrose inequality in Minkowski was proved
is due to Gibbons \cite{Gibbons1997}, who considered convex surfaces $\smink$
lying on a spacelike hyperplane and established the Penrose inequality with respect
to the Killing orthogonal to the hyperplane. This case is immediately
covered by Theorem \ref{ThWang}.  In fact, this theorem
also implies the validity of the Penrose inequality for $\smink$
with respect to {\it any} other time translation, as we show next.
\begin{Tma}
Let $\smink$ be a closed, connected and convex surface embedded in a spacelike hyperplane
$\Sigma'_{t'_0}\hookrightarrow \mathcal{M}^{1,3}$. Let 
$\xi$ be any unit time translation (not necessarily orthogonal
to $\Sigma'_{t_0}$). Then the Penrose inequality with respect to $\xi$ holds
for  $\smink$.
\end{Tma}
\begin{proof}
Let $\nu'$ be the outward normal to $\smink$ in $\Sigma'_{t'_0}$.
Since a hyperplane is totally geodesic, the second fundamental form
vector of $\smink$ is $K = K^{\nu'} \nu'$, where $K^{\nu'}$ is 
positive semidefinite. Choose any hyperplane $\Sigma_{t_0}$ orthogonal to
$\xi$ and define $\suptres$ as the orthogonal projection 
of $\smink$ onto $\Sigma_{t_0}$. To prove the theorem it suffices
to show that $\suptres$ is convex, i.e. that its second fundamental
form $\overline{K}$ with respect to the unit outer normal $\overline{\nu}$
in $\Sigma_{t_0}$ is non-negative. From Proposition \ref{proj} in the
Appendix (with $V=1$, as we are in Minkowski) we have
\begin{eqnarray*}
K^{\overline{\nu}}  = \pi^{\star} (\overline{K}),
\end{eqnarray*}
where $\pi : \smink \rightarrow \suptres$ is the projection along
$\xi$, $\overline{\nu}$ is the parallel extension along $\xi$ of the
normal vector $\overline{\nu}$ of $\suptres$ evaluated on $\smink$
and $K^{\overline{\nu}} \defi \langle K, \overline{\nu} \rangle_{\eta}$.
Thus, $\overline{K}$ is non-negative if and only if 
$\langle \nu', \overline{\nu} \rangle_{\eta}$ 
is non-negative. Now, both
$\nu'$ and $\overline{\nu}$ are normal to $\smink$, spacelike and unit. Since they
belong to a two-dimensional Lorentzian space,
$\langle \nu', \overline{\nu} \rangle_{\eta}$ vanishes nowhere,
and, hence, has constant sign. For the choice $\xi=\xi'$, i.e. the time 
translation normal to $\Sigma'_{t_0}$, we obviously have
$\nu'  = \overline{\nu}$ and the sign is positive. Since $\xi$ can be obtained
from $\xi'$ by a smooth deformation, and $\overline{\nu}$ also changes smoothly,
it is impossible that the sign of $\langle \nu',\overline{\nu} \rangle_{\eta}$
changes from $+1$ to $-1$, and the theorem is proved. 
\end{proof}

This theorem implies a Minkowski type inequality for
$S' \defi \smink$ as a convex surface of Euclidean space. Indeed, 
the Killing vector $\xi$ can be decomposed as 
$\xi = \sqrt{1 + |v|^2} \, \xi' + v$ where $v$
is a translation of Euclidean space $(\mathbb{E}^{3},\gE)$ (identified with
the hyperplane $\Sigma^{'}_{t'_0}$). With the definition of null
vectors $k = \xi' - \nu'$ and $\ell = \xi' + \nu'$ on $S'$ and, given that
the mean curvature vector of $S'$ is  $H' \nu'$, where $H'$ is the mean curvature of
$S' \hookrightarrow \mathbb{E}^3$, the
Penrose inequality (\ref{PI}) with respect to $\xi$ 
becomes
\begin{equation}
\label{Minkgeneral}
\int_{S'} H'f\bm{\eta_{S'}}\geq \sqrt{16\pi|S'|},
\end{equation}
where 
\begin{equation*}
f=\sqrt{1+|v|^2}+\gE(\nu',v ). 
\end{equation*}
Obviously, when $v=0$ we recover the standard Minkowski inequality.
The validity of this inequality suggests that it might be
worth studying for which functions
$f$ Minkowski type inequalities of the form (\ref{Minkgeneral}) hold
for arbitrary convex surfaces of Euclidean space. We note in this respect
different, but somewhat related results in \cite{Miao}.

\section{Appendix}

One of the ingredients in Theorem \ref{ThWang} of Brendle and Wang is 
a computation relating the extrinsic curvatures of a
codimension-two
spacelike surface $S$ embedded in a strictly static spacetime
and its projection $\suptres$ onto
a hypersurface of constant static time.  A similar and more exhaustive
analysis in the case of Minkowski spacetime was carried out
in connection with a new definition of quasi-local mass in
\cite{WangYau1}, \cite{WangYau2}. Results of this type in the general
static case have also appeared in \cite{BrayKhuri}. However, to the best of our knowledge,
no systematic account of the relation between all the intrinsic
and extrinsic geometric properties of $S$ and its projection $\suptres$
has appeared in the literature, neither in the Minkowski nor in the general
strictly static case. We devote this Appendix to doing so. 

Let $(\M,g)$ be an
$(n+2)$-dimensional spacetime with a
Killing vector field $\xi$ which is everywhere timelike and hypersurface
orthogonal. We choose $\xi$ to be future directed.
The covariant derivative of $(\M,g)$ is denoted as $\nabgen$
and $\langle \cdot,\cdot \rangle$ is used for scalar products with $g$. 
The norm $V>0$ of the Killing vector is defined by 
$\langle \xi, \xi \rangle=- V^2$. 

Consider a codimension-two spacelike surface $S$ in $\M$. 
Since all calculations
are local we can assume without loss of generality
that $S$ is embedded,  and that there exists a time
function $t : \M \rightarrow \mathbb{R}$ such that $\bm{\xi} = - V^2 dt$ (this
follows from the integrability of $\bm{\xi}$ and locality). Choose
any $t_0 \in \mathbb{R}$ and let $\Sigma_{t_0} = \{ t = t_0 \}$. The projection
$\suptres$ of $S$ onto $\Sigma_{t_0}$ along the orbits of $\xi$
defines a codimension-two surface which again can be taken to be embedded
(after restricting $S$ if necessary). Thus, we have a diffemorphism
$\pi : S \rightarrow \suptres$ defined by projection along $\xi$.
The induced metrics and covariant derivative on $S$ (resp. $\suptres$)
are denoted as $\gamma$ and $D$ (resp. $\overline{\gamma}$ and 
$\overline{D}$). The function $\tau \defi t|_S - t_0$  and
$\VS \defi V|_S$ will play an important role in relating
the geometry of the two surfaces. As before, scalar functions on $S$ will be transferred
to $\overline{S}$ by means of $\pi$ while keeping their names. The precise
meaning will follow from the context.

For any vector field $X \in \mathfrak{X} (S)$ we denote its projection
$d\pi (X) \in \mathfrak{X} (\overline{S})$ as $\overline{X}$. Given any
such vector $\overline{X}$ 
we extend it along the orbits of
the Killing vector by Lie transport along $\xi$, i.e. solving
$[\xi,\overline{X} ]=0$.  Again we keep the same name for the extension. Note
that $\overline{X}$ is everywhere orthogonal to $\xi$. With these definitions
it is straightforward that, at any $p \in S$,
\begin{equation}
\label{vectrelation}
X|_p=\overline{X}(\tau)\xi|_p+\overline{X}|_p.
\end{equation}
As a consequence,
the metrics $\gamma$ and
$\overline{\gamma}$ are related by
\begin{align*}
\gamma(\vecuno,\vecunoy)|_p & =\langle \overline{\vecuno}(\tau)\xi+\overline{\vecuno}, \overline{\vecunoy}(\tau)\xi+\overline{\vecunoy}\rangle|_p=
\left ( \pi^*(\mettres)
-\VS^2 d\tau\otimes d\tau|_p \right ) (\vecuno,\vecunoy)|_p 
\end{align*}
where we have used  $d \pi|_p (X) = \overline{X} |_{\pi(p)}$
and $\overline{X} (\tau) = d \tau (X)$. So, we conclude
\begin{eqnarray*}
\gamma  = \pi^{\star} (\overline{\gamma}) - \VS^2 d \tau \otimes d \tau.
\label{metrics}
\end{eqnarray*}
The inverse metrics are then related by
\begin{eqnarray*}
\overline{\gamma}^{-1}=d\pi(\gamma^{-1})-
\frac{\VS^2}{\W^2}\grad_{\,\mettres}(\tau)\otimes\grad_{\,\mettres}(\tau),
\quad \quad
\W \defi \sqrt{1-\VS^2|d\tau|^2_{\mettres}}, 
\end{eqnarray*}
which has, as immediate consequences,
\begin{equation}
\label{consecmetricainv}
d\pi (\grad_{\,\gamma}(\tau))=\frac{1}{\W^2}\grad_{\,\mettres}(\tau), 
\quad |d\tau|^2_{\gamma}=\frac{|d\tau|^2_{\mettres}}{\W^2}, 
\quad
 \W=\frac{1}{\sqrt{1+\VS^2|d\tau|^2_{\gamma}}}.
\end{equation}
The bound $1-\VS^2|d\tau|^2_{\mettres}>0$  (necessary for 
$\W$ to be real) is a consequence of
$S$  being spacelike everywhere.
It is also immediate to show that the respective volume forms
$\bm{\eta_{S}}$ and $\bm{\eta_{\suptres}}$ 
are related by
\begin{equation}
\bm{\eta_{S}}=\W \bm{\eta_{\suptres}}.
\end{equation}

In order to study the relation between the extrinsic geometries of $S$ and
$\suptres$ it is useful to choose a basis of the normal bundle
of each surface. Concerning $\suptres$, the natural choice is
$\{ \ndos, \VS^{-1} \xi |_S \}$, where $\ndos$ is a unit normal
of $\suptres$ as a hypersurface in $\Sigma_{t_0}$. We denote by
$\overline{K}$ the second fundamental form of $\suptres$ along
$\ndos$. Concerning $S$, the Lie constant extension $\ndos$ along the Killing 
$\xi$ defines a spacelike and unit normal to $S$, still denoted
by $\ndos$. For the second vector, note that $\xi |_S$ is nowhere 
tangent to $S$ and hence its normal component $\xi^{\bot}$
in the orthogonal decomposition $T_p \M = T_p S \oplus N_pS$ is nowhere
zero and, in fact, timelike. From $\bm{\xi} = - V^2 dt$ we have,
for any $X\in T_p S$,  $\langle \xi|_S,X\rangle=-\VS^2 d\tau(X)$ which means
that the tangential component of $\xi|_S$ is $-\VS^2\grad_{\,\gamma}(\tau)$, or
equivalently $\xi^{\bot} = \xi |_S  + \VS^2 \grad_{\,\gamma}(\tau)$.
Following \cite{WangYau2} we denote 
by $u$ the future directed unit vector tangent to $\xi^{\bot}$. Its
explicit form is
\begin{eqnarray}
u = \frac{\W}{\VS} \left ( \xi|_S + \VS^2 \grad_{\,\gamma} (\tau) \right )
\label{e4}
\end{eqnarray}
as a consequence of $u$ being unit and orthogonal to 
$\grad_{\,\gamma} (\tau)$ and the property $\langle \xi, \xi \rangle = - V^2$.
We note that $\{ \overline{\nu}, u \}$ defines an orthonormal basis 
of the normal bundle of $S$.

The extrinsic geometry of $S$ is encoded into its second fundamental form
vector $K$ and the connection of the normal bundle $\bm{\alpha}$.
For the basis above, this geometric information is in turn given by
the two symmetric tensors $K^{u} \defi \langle K,u \rangle$,
$K^{\ndos} \defi \langle K, \ndos \rangle$ and the one-form
$\bm{\alpha_{\ndos}}(X) \defi \langle\nabgen_{X}\ndos,u\rangle$, $X \in 
\mathfrak{X}(S)$. The following proposition relates
these
objects with the geometry of the projected surface:
\begin{Prop}
\label{proj}
With the notation above,
\begin{eqnarray}
\label{kmtensor}
K^{\ndos}&=&\pi^*(\overline{K})-\VS\ndos(V) |_Sd\tau\otimes d\tau, \\
\label{ke4tensor}
K^{u}&=&\frac{1}{\W}\left(d\VS\otimes d\tau+d\tau\otimes d\VS+\VS\pi^*(\Hess_{\,\mettres}(\tau) ) \right)
-\frac{\VS^2}{\W}d\VS(\grad_{\,\mettres}(\tau) )d\tau\otimes d\tau,  
  \\
\label{conform}
\bm{\alpha_{\ndos}}&=&\frac{1}{\W}\left(
\VS\pi^{\star}(\overline{K}(\grad_{\,\mettres}(\tau),\cdot))-\ndos(V)|_S d\tau\right) \hspace{1cm} \W=\sqrt{1-\VS^2|d\tau|^2_{\mettres} }, 
 \end{eqnarray}  
\end{Prop}
\begin{proof}
Inserting (\ref{vectrelation}) in the defining expression
$K^{\ndos}(\vecuno,\vecunoy)=\langle\nabgen_{\vecuno}\ndos,\vecunoy\rangle$
gives, after using $\overline{X}(\tau)=d\tau(X)$,
\begin{eqnarray}
K^{\ndos}(\vecuno,\vecunoy)
=d\tau(Y)\langle\nabgen_{X}\ndos,\xi\rangle+d\tau(X)\langle\nabgen_{\xi}\ndos,\overline{Y}\rangle+\langle\nabgen_{\overline{X}}\ndos,\overline{Y}\rangle.
\label{inter1}
\end{eqnarray}
Now, $\langle\nabgen_{X}\ndos,\xi\rangle=
\overline{X}(\tau)\langle\nabgen_{\xi}\overline{\nu},\xi\rangle
+\langle\nabgen_{\overline{X}}\overline{\nu},\xi\rangle
= \overline{X}(\tau)\langle\nabgen_{\xi}\overline{\nu},\xi\rangle$,
the second equality following from  $\Sigma_{t_0}$ being totally geodesic.
To elaborate this further, we note that
$d \bm{\xi} = 2 V^{-1} d V \wedge \bm{\xi}$ as a
consequence of $\bm{\xi} = - V^2 dt$. Hence
\begin{eqnarray}
\nabgen_{\overline{\nu}} \bm{\xi}  = \frac{1}{2} d \bm{\xi}  ( \overline{\nu},
\cdot ) =  \frac{\overline{\nu} (V) |_S}{\VS} \bm{\xi},
\end{eqnarray}
where in the first equality we used the Killing equations
and in the second the orthogonality of $\overline{\nu}$ and 
$\xi$.
Raising indices and recalling that
$[\xi,\overline{\nu}] =0$ we conclude 
\begin{eqnarray}
\label{extnormalkilling}
\nabgen_{\xi} \overline{\nu} = 
\nabgen_{\overline{\nu}} \xi = \frac{\overline{\nu} (V)|_{S}}{\VS} \xi,
\end{eqnarray}
and therefore
\begin{equation}
\label{KNprimerterm}
\langle\nabgen_{X}\ndos,\xi\rangle=-\VS\ndos(V)|_S d\tau(X).
\end{equation}

With these expressions at hand,
the first term in (\ref{inter1})  becomes
$-\VS \ndos(V)|_S (d\tau \otimes d\tau) (X,Y)$,
while the second term vanishes. Finally, the last term
gives the second fundamental form of $\suptres$ and (\ref{kmtensor})
follows (to our knowledge, this identity  appeared for the first time
in \cite{BrendleWang2013}).

Concerning $K^{u}$, its symmetry properties allows us to write
$K^{u}( \vecuno,\vecunoy)= \frac{1}{2} (\langle\nabgen_{\vecuno}u,\vecunoy\rangle
+ \langle\nabgen_{\vecunoy}u,\vecuno\rangle )$, which after inserting 
(\ref{e4}) yields
\begin{align*}
\label{ke4uno}
K^{u}(\vecuno,\vecunoy) = &
%=\frac{1}{2}\left (\langle\nabgen_{\vecuno}[\cnulldos(\xi+V^2\grad_\gamma(\tau))],\vecunoy\rangle+\langle\nabgen_{\vecunoy}[\cnulldos(\xi+V^2\grad_\gamma(\tau))],\vecuno\rangle  \right)= \\
\frac{\W}{2\VS}\left(\langle\nabgen_{\vecuno}\xi,\vecunoy\rangle+\langle\nabgen_{\vecunoy}\xi,\vecuno\rangle+\langle\nabgen_{\vecuno}(\VS^2\grad_{\,\gamma}(\tau)),\vecunoy\rangle+\langle\nabgen_{\vecunoy}(\VS^2\grad_{\,\gamma}(\tau)),\vecuno\rangle  \right).
\end{align*}
The Killing equations imply that the first two terms cancel each other.
Expanding the remaining terms  it follows immediately
\begin{equation}
\label{Ke4geoms}
K^{u}=\W \left (d\VS\otimes d\tau+d\tau\otimes d\VS+
\VS\Hess_{\,\gamma}(\tau)\right).
\end{equation}
In order to rewrite this in terms of the projected geometry, 
we need to find the relation between the Hessians of $\tau$
on each one of the surfaces. To that aim, recall
that the difference between connections $D$ and $\overline{D}$
on a given manifold defines a type $(1,2)$ tensor  $\mathcal{Z}$ 
such that the following identity 
holds for any one-form $\bm{\omega}$
(see e.g. \cite{Wald}):
\begin{equation}
\label{difunoformas}
(D\bm{\omega})(X,Y)-(\overline{D}\bm{\omega})(X,Y)=
-\mathcal{Z}(\bm{\omega},X,Y).
\end{equation}
In our context, we can use $\pi^{\star} (\overline{\gamma}) $ on $S$
and the corresponding connection $\overline{D}$ it defines.
Given the relation (\ref{metrics}), a straightforward computation
gives 
\begin{equation*}
\mathcal{Z}(d\tau,\cdot,\cdot)=-\VS|d\tau|^2_\gamma\left( d\VS\otimes d\tau+d\tau\otimes d\VS+\VS\pi^*(\Hess_{\,\mettres}(\tau) ) \right)+\VS d\VS(\grad_{\,\gamma}(\tau))d\tau\otimes d\tau
\end{equation*}
Inserting this into (\ref{difunoformas}) with $\bm{\omega} 
\rightarrow d \tau$ and using (\ref{consecmetricainv}) 
 it follows
%\begin{eqnarray*}
%& &\pi^*(\Hess_{\,\mettres}\tau)-\Hess_{\metdos}\tau= \\
%& &=-\frac{|d\tau|^2_{\mettres}}{\cnulldos^2V}(dV\otimes d\tau+d\tau\otimes dV)%+\frac{dV(\grad_{\,\mettres}\tau)}{\cnulldos^2V}d\tau\otimes d\tau-\frac{|d\tau|^%2_{\mettres}}{\cnulldos^2}\pi^*(\Hess_{\,\mettres}\tau).
%\end{eqnarray*} 
%Considering the third expression from (\ref{consecmetricainv}) and using it in %the above expression, we finally obtain the relation between the hessians:
\begin{equation}
\label{hessrelation}
\W^2\Hess_{\,\metdos}(\tau)=\pi^*(\Hess_{\,\mettres}(\tau) )+
\VS|d\tau|^2_{\overline{\gamma}}(d\tau\otimes d\VS+d\VS\otimes d\tau)-
\VS d\VS (\grad_{\,\mettres}(\tau) ) d\tau\otimes d\tau. 
\end{equation}
Combining this and (\ref{Ke4geoms}) gives (\ref{ke4tensor}) at once.

It only remains to compute 
the connection 1-form $\alpha_{\ndos}(X)=\langle\nabgen_{X}\ndos,u\rangle$.  
Substituting (\ref{e4}) and recalling that $\overline{\nu}$ is orthogonal
to $u$ one finds
\begin{align*}
%\label{connectionform}
 \alpha_{\ndos}(X)& =\frac{\W}{\VS}\langle\nabgen_{X}\ndos,\xi+\VS^2\grad_{\,\gamma}(\tau)\rangle=\frac{\W}{\VS}\langle\nabgen_{X}\ndos,\xi\rangle
+\W \VS K^{\ndos}(\grad_{\,\gamma}(\tau),X) \\
& = 
-\W\ndos(\VS)d\tau(X) +\W \VS K^{\ndos}(\grad_{\,\gamma}(\tau),X),
\end{align*} 
where in the last equality we used (\ref{kmtensor}). Replacing
(\ref{KNprimerterm}) and using the first relation
in (\ref{consecmetricainv}) and the definition of $\W$,
then 
\begin{eqnarray*}
\alpha_{\ndos}(X)&=&-\W\ndos(\VS)d\tau(X)+\W\VS \left(\pi^*(\overline{K})(\grad_{\,\gamma}(\tau),X)-\VS\ndos(V) |_S|d\tau|^2_\gamma d\tau(X)  \right)= \\
&=&\frac{1}{\W}\left(
\VS\overline{K}(\grad_{\,\mettres}(\tau),d\pi(X))-\ndos(V)|_S d\tau(X)\right) 
\end{eqnarray*}
as claimed.

\end{proof}

\noindent {\bf Remark.}  Although we have assumed $\xi$ to be timelike, all
the calculations in the Appendix are similar when $\xi$ is spacelike
and nowhere zero.
In particular the geometric relations between $S$
and its projection $\overline{S}$ in a purely Riemannian context where $\langle\xi,\xi\rangle=V^2$  and $\bm{\xi} = V^2 dt$ are
\begin{eqnarray*} 
\gamma & = &\pi^{\star} (\overline{\gamma}) + \VS^2 d \tau \otimes d \tau, \\
\bm{\eta_{S}}&=&\W \bm{\eta_{\suptres}}  \hspace{1cm} \W=\sqrt{1+\VS^2|d\tau|^2_{\mettres} }, \\
K^{\ndos}&=&\pi^*(\overline{K})+\VS\ndos(V) |_Sd\tau\otimes d\tau, \\
K^{u}&=&-\frac{1}{\W}\left(d\VS\otimes d\tau+d\tau\otimes d\VS+\VS\pi^*(\Hess_{\,\mettres}(\tau) ) \right)
-\frac{\VS^2}{\W}d\VS(\grad_{\,\mettres}(\tau) )d\tau\otimes d\tau,  
  \\
\bm{\alpha_{\ndos}}&=&\frac{1}{\W}\left(
- \VS\pi^{\star}(\overline{K}(\grad_{\,\mettres}(\tau),\cdot))+
\ndos(V)|_S d\tau\right), 
 \end{eqnarray*}  
where this time the unit vector $u$ reads
\begin{eqnarray*}
u = \frac{\W}{\VS} \left ( \xi|_S - \VS^2 \grad_{\,\gamma} (\tau) \right ).
\label{e4}
\end{eqnarray*}

\noindent {\bf Remark.} Note that the expressions above contain all the information
needed to relate any geometric quantity on $S$ with geometric information on 
its projection $\suptres$. For instance, the mean curvature vector
of $S$ can be related to the projected geometry simply taking
the trace in $K = K^{\overline{\nu}} \overline{\nu} - K^{u} u$
with the metric $\gamma^{-1}$ and using 
(\ref{metrics}) together with the results in Proposition \ref{proj}.
Similarly, the null second fundamental forms 
$K^{k}$, $K^{\ell}$ of $S$ along a basis of null 
normals $\{k, \ell\}$ can be obtained directly from Proposition \ref{proj}
after decomposing $\{k,\ell\}$
in the basis $\{\overline{\nu},u\}$. The same applies 
to the corresponding null expansions.  

Concerning the connection one-form, its behaviour under change of basis
is not tensorial (being a connection), so it may be worth
giving its explicit expression in a null-basis $\{k,\ell\}$ of the form
$k=f (- \ndos+ u) $ and $\ell= f^{-1} (\ndos+ u)$ where $f: S \rightarrow
\mathbb{R} \setminus \{0\}$ is smooth. With the usual definition
of connection one-form in this basis given by
$\bm{s}(X)\defi \frac{1}{2}\langle\nabgen_{X}k,\ell\rangle$ we have
\begin{equation*}
\bm{s}(X)=\frac{1}{2}\langle\nabgen_{X}k,\ell\rangle=\frac{1}{2}\langle
\nabgen_{X}(-f\ndos+fu), f^{-1} \ndos+ f^{-1} u\rangle=-\frac{X(f)}{f}-\alpha_{\ndos}(X),
\end{equation*}  
and hence 
\begin{equation*}
\bm{s} = - \frac{df}{f} + \frac{1}{\W} \left ( \overline{\nu} (V)|_{S}
  d\tau - \VS\pi^{\star} (\overline{K}(\mbox{grad}_{\,\mettres}(\tau), \cdot ) )
\right ).
\end{equation*}

\section{Acknowledgments}
Financial support under the projects  FIS2012-30926 (MICINN)
and P09-FQM-4496 (Junta de Andaluc\'{\i}a and FEDER funds)
are acknowledged. A.S. acknowledges the Ph.D. grant AP2009-0063 (MEC).

\end{document}